\documentclass[12pt, a4paper]{article}
\usepackage{amsmath,amssymb,amsthm,graphicx,color,bbm,array}
\usepackage[left=1in,right=1in,top=1in,bottom=1in]{geometry}
\usepackage{cite}
\usepackage[british]{babel}
\usepackage{hyperref} 
\usepackage{amsfonts}
\usepackage{verbatim}
\usepackage{graphicx}
\usepackage{color}
\usepackage{enumerate}
\usepackage[utf8]{inputenc}
\usepackage[T2A]{fontenc}

\usepackage{hyperref}

\newcommand{\mathsym}[1]{{}}

\newcommand{\unicode}[1]{{}}

\newtheorem {lemma}{Lemma}
\newtheorem {theorem}{Theorem}
\newtheorem{remark}{Remark}

\newcommand{\dd}{\mathrm{d}}

\newcommand{\Q}{\mathbb{Q}}

\newcommand{\QT}{\mathbb{Q}_{\rm T}}

\newcommand{\T}{{\rm T}}
\newcommand{\E}{{\mathbb{E}}}
\newcommand{\Var}{\mathsf{Var}}
\newcommand{\Cov}{\mathsf{Cov}}

\newcommand{\R}{\mathbb{R}}

\newcommand{\eps}{\varepsilon}

\newcommand{\A}{{\cal A}}

\newcommand{\e}{{\bf e}}

\begin{document}
\title{
Explicit local volatility formula for Cheyette-type   interest rate models
}
\author{Alexander Gairat,\footnote{Sberbank, Moscow, Russia. Email address: asgayrat@sberbank.ru}
Vyacheslav Gorovoy\footnote{
	New Economic School, Moscow, Russia.
	Email address: vgorovoy@nes.ru
}\, and 
Vadim Shcherbakov\footnote{
	Royal Holloway,  University of London, Egham, UK.
	Email address: vadim.shcherbakov@rhul.ac.uk
}
}

\date{}
\maketitle


\abstract{
{\small    
This paper addresses the approximation of the local volatility function in the Cheyette interest rate model.
Its main contribution is an explicit analytical formula for approximating local volatility, derived by extending the classical Dupire framework to interest rate models. In particular, an implicit Dupire-like expression for local volatility is first derived  for options written on the short rate. This expression is then approximated using a combination of perturbation methods and probabilistic techniques, resulting in a 
formula expressed in terms of time and strike derivatives of the Bachelier implied variance.
The final formula naturally extends to multi-factor Cheyette models and provides a
 practical tool for model calibration.
}
}

\bigskip


\section{Introduction}
\label{intro}
 The Cheyette model and its modifications  are well known and widely used
  by both practitioners and researchers.
  These models are  valued for 
  their mathematical tractability, which enables their  efficient numerical implementation.
The original Cheyette model is a single-factor quasi-Gaussian HJM model
 with the time dependent deterministic diffusion coefficient (volatility). 
 A known limitation of the standard Cheyette model, similar to that of basic models
  in equity and foreign exchange (FX) markets (e.g., the Black–Scholes model),
   is its inability to capture market smiles and skews in implied volatilities.
 One way  to deal with this drawback of basic models 
 is to consider their  local volatility extensions, where 
the volatility is  a function of both the  time and the state variable.

Local volatility models offer a practical and arbitrage-free framework for 
capturing the implied volatility smile
 and skew observed in the market. A widely used tool for reconstructing the local volatility surface is Dupire’s 
 formula, which relies on partial derivatives of option prices with respect to strike and maturity.
Dupire-like implicit formulas for local volatility in interest rate models have been studied previously
 (see, e.g., Cao and Henry-Labordère~\cite{CL}, Gatarek et al.~\cite{GJQ}, Lucic \cite{L}, Savine~\cite{Savine}
  and references therein). 
 This paper contributes to the literature by deriving an explicit analytical formula for approximating the local 
 volatility function in the Cheyette model.

 It should be noted that applying Dupire’s classical framework to fixed-income markets is not straightforward. 
Unlike equity and FX markets, where options on the same underlying asset exist across various maturities, 
fixed-income markets typically feature instruments known as rolling maturity options, that is, options whose 
underlying rate instruments evolve over time as their maturities change. For example, swaptions with different 
maturity profiles, even if based on the same tenor swap, effectively have different underlying swap rates.

The main results of this paper are formulated for options on the short rate and involve the corresponding short-rate implied total variance surface. Although such options are not traded, they provide a natural and convenient framework for constructing a local-volatility representation.
In practice, however, the interest-rate options market primarily consists of swaptions rather than short-rate 
options. 
Attempting to construct an analogue of Dupire’s formula directly from swaption smiles would require 
differentiating prices of payoffs that depend nonlinearly on the entire bond curve at expiry and, in the Cheyette 
model, on the two-dimensional state of the underlying process, which makes 
the resulting problem analytically intractable.   
Introducing short-rate options as an intermediate (model-implied) object allows us 
to obtain an explicit analytic extension of Dupire local volatility to the interest-rate setting and develop 
a relatively simple  swaption-based calibration method described in Section~7. 
Thus, working with  options on the short rate  is not merely 
 a modeling convenience but a key step of our analysis 
 that enables analytical tractability.
To the best of our knowledge, this approach has not been considered previously.

The structure of the paper is as follows. In Section~\ref{Model}, we
 introduce the one-factor Cheyette model and establish the key notations. 
    The main results are presented in Section~\ref{main-res}, and their proofs are 
 provided in Section~\ref{sec:Proofs}.
 Section~\ref{SecNFactorCase} extends the 
 framework to multi-factor Cheyette models, with a detailed example of the two-factor case discussed in 
 Section~\ref{2-factor}. The application of the 
 proposed method 
 for calibrating the Cheyette model to swaptions market
 is described in Section~\ref{swaption:calibration}. Finally, Section~\ref{Gaussian} provides
  computational details for the two-factor Gaussian case.

\section{One-factor Cheytte model}
\label{Model}

We begin by recalling the one-factor Cheyette model (see, e.g., Andersen 
and Piterbarg~\cite{AndersenPiterbarg2}) and introducing key notations used throughout the paper.

Let 
 $x_t=(x_t,\, t\geq 0)$ and $y_t=(y_t,\, t\geq 0)$ be processes  that 
 satisfy equations
\begin{align*}
dx_t &= (y_t - \mu x_t)\,dt + \sigma(t,x_t)\,dW_t, \\
dy_t &= \big(\sigma^2(t,x_t) - 2\mu y_t\big)\,dt, \\
x_0 &= y_0 = 0,
\end{align*}
where $W_t=(W_t,\, t\geq 0)$ is a standard Brownian motion under the risk-neutral 
measure $\Q$, and $\sigma(t,x)$ is a deterministic function of time and space.
In the  one-factor Cheyette model 
the instantaneous forward rate $f_t(T)$ for maturity $T \geq t$,
 and the short rate $r_t$, are the following functions of  $x_t$ and $y_t$
\begin{align}
\label{f(t,T)}
f_t(T) &= f_0(T) + e^{-\mu (T-t)}\big(x_t + G(t,T) y_t\big), \quad 0\leq t\leq T, \\
\label{r}
r_t &= f_0(t) + x_t, \quad t\geq 0,
\end{align}
where the function $G(t,T)$ is defined by
\begin{equation}
\label{G}
G(t,T) = \frac{1 - e^{-\mu (T-t)}}{\mu}, \quad 0 \leq t \leq T.
\end{equation}

Throughout, we assume all processes are adapted to a filtration $({\cal F}_t,\, t\geq 0)$.
Given a probability measure $M$ and a random variable $X$, we denote $\E^M_t(X) = \E^M(X \mid {\cal F}_t)$ and $\E^M(X) = \E^M_0(X)$.

Under these notations, the time-$t$ price of a zero-coupon bond maturing at time $T$ is given by
\begin{equation*}
\begin{split}
P_t(T) &= \E_t^{\Q}\left( e^{-\int_t^T r_u \, du} \right) 
= \frac{P_0(T)}{P_0(t)} \exp\left( -G(t,T)x_t - \frac{1}{2}G^2(t,T)y_t \right), \quad 0 \leq t \leq T.
\end{split}
\end{equation*}
Under the $\T$-forward measure $\QT$, defined by
\[
\frac{d\QT}{d\Q} = \frac{e^{-\int_0^T r_u\, du}}{P_0(T)},
\]
the process $f_t(T)$
follows  the equation
\begin{align}
\label{f-SDE-T}
df_t(T) &= \sigma_T(t, x_t)\,dW_t^{\T},
\end{align}
where
\begin{equation}
\label{sigma_T}
 \sigma_T(t, x) = e^{-\mu(T-t)}\, \sigma(t, x),
\end{equation}
and $W^{\T}_t$ denotes a standard Brownian motion under $\QT$.

Next, introduce the European option on the short rate. 
Specifically, 
consider a European call option with maturity $T$ and strike $K$, written 
on $r_T = f_T(T)$. Its time-$0$ price is given by
\begin{equation}
\label{option_price}
P_0(T)\E^{\T}\left( (r_T - K)_+\right),
\end{equation}
where $\E^{\T}:=\E_0^{\QT}$.
Recall that under 
Bachelier's model with volatility $\sigma$ 
the  {\it non-discounted price} of the option is given by
\begin{equation}
\label{BH}
\text{BH}(f_0(T),K,T,\sigma) = (f_0(T) - K)\,\Phi\left(\frac{f_0(T)-K}{\sigma\sqrt{T}}\right) + \sigma\sqrt{T}\,\phi\left(\frac{f_0(T)-K}{\sigma\sqrt{T}}\right),
\end{equation}
where $\Phi$ and $\phi$ denote the cumulative distribution function and probability 
density function, respectively, of the standard normal distribution $\mathcal{N}(0,1)$.

The Bachelier implied volatility $\sigma_{\text{imp}}(T,K)$ on short rate option is defined as the solution to
\begin{align*}
C(T, K):=
\E^{\T}\left( (r_T - K)_+\right) &= \text{BH}\left(f_0(T),K,T,\sigma_{\text{imp}}(T,K)\right).
\end{align*}
Introduce the shifted forward process
$\widetilde{f}_t(T) = f_t(T) - f_0(T)$, $\widetilde{r}_T =\widetilde{f}_T(T)$,
noting that $\widetilde{f}_0(T) = 0$. 
Let $k = K - f_0(T)$ and rewrite 
 the aforementioned non-discounted option price $C(T, K)$ in these terms as follows
\begin{equation}
\label{C_t(T,k)}
C(T, k)=
\E^{\T}\left((\widetilde{r}_T - k)_+\right)=
\text{BH}\left( 0, k, T, v(T,k) \right),
\end{equation}
where 
$v(T,k) = \sigma_{\text{imp}}(T,k+f_0(T))$.
The total implied variance is
\begin{equation}
\label{w}
w(T,k) := T v^2(T,k).
\end{equation}
\begin{remark}
\label{BH(w)}
{\rm 
It is convenient to view the Bachelier price BH as a function of the total implied variance, 
rather than the implied volatility. Accordingly, with a slight abuse of notation, we will also write
\begin{equation}
\label{C_t(T,k)-1}
C(T, k)=\text{BH}\left(k, w(T,k) \right).
\end{equation}  
}
\end{remark}
Without loss of generality, we assume $f_0(T) = 0$ for the remainder of the paper, 
and  simplify notation by omitting the tilde in $\widetilde{f}_t(T)$.
Under this assumption, and in view of \eqref{f(t,T)} and \eqref{r}, 
we have that 
\begin{equation}
\label{f(t,T)-0}
f_t(T) = e^{-\mu(T-t)}\left(x_t + G(t,T) y_t\right),
\quad
r_T = f_T(T) = x_T.
\end{equation}

Finally, note that  we use throughout standard 
notations $\partial_k = \frac{\partial}{\partial k}$, 
$\partial_{kk} = \frac{\partial^2}{\partial k^2}$, $\partial_T = \frac{\partial}{\partial T}$ etc
for partial derivatives.

\section{Results}
\label{main-res}

\subsection{The main result}
\label{main}

The main result of this paper is the following analytical approximation for the local volatility function 
\begin{equation}
\label{main-1-factor}
\sigma^2(T, k) \approx \frac{ \partial_T w + \mu \left( 2w - k\,\partial_k w \right) + w\,\partial_k w }
{\left(1 - \frac{k\,\partial_k w}{2w} \right)^2 +
 \frac{1}{2}\left( \partial_{kk}w - \frac{(\partial_k w)^2}{2w} \right)}
+ (\partial_k w)^3,
\end{equation}
where $w = w(T,k)$ is the implied total variance defined in~\eqref{w} and $\partial_kw=\partial_kw(T,k)$.
    Numerical experiments demonstrate  (see, e.g.  Figure~\ref{fig:IV} below) that the approximation fits 
  market smiles with minimal calibration error,  offering 
  a practical and efficient tool for modeling volatility smiles in interest rate markets.
 The approximation 
 extends naturally 
to the multi-factor Cheyette model, where a similar form holds with the mean-reversion
 parameter $\mu$ replaced by an effective mean reversion $\mu_{\text{eff}}(T)$
 (see Section~\ref{SecNFactorCase} for more details).

The approximating formula~\eqref{main-1-factor} is
 justified by Theorems~\ref{ImplicitFormula},~\ref{thm:first_order} and~\ref{thm:linear}
 presented  below.

\subsection{Implicit formula for  local volatility}
\label{implicit}

We begin with the following result.
\begin{theorem}
\label{ImplicitFormula}
The local volatility $\sigma^2(T,k)$ satisfies the equation
\begin{equation}
\label{genericLV}
\sigma^2(T,k) = 2\, \frac{
\partial_T C(T,k) + \mu\left( C(T,k) - k\,\partial_k C(T,k) \right) + \E^{\T}\left(x_T (x_T-k)_+ \right) -
 \E^{\T}\left(y_T\,\theta(x_T-k) \right)
}
{ \partial_{kk} C(T,k) },
\end{equation}
where  $C(T, k)$ is the non-discounted option price given by~\eqref{C_t(T,k)} and 
\begin{equation}
\label{theta}
\theta(x) = 
\begin{cases}
1, & x \geq 0, \\
0, & x < 0.
\end{cases}
\end{equation}
\end{theorem}
Theorem~\ref{ImplicitFormula} can be derived from limiting cases 
of swap rate dynamics in Cao and Henry-Labordère~\cite{CL}, Gatarek et al.~\cite{GJQ} and Lucic~\cite{L}.  
An alternative proof is provided in Section~\ref{proofImplicit}.

\begin{remark}
{\rm
It should be noted that the formula in 
Theorem~\ref{ImplicitFormula} is implicit in the following sense.
The term $\E^{\T}\left(y_T\theta(x_T-k) \right)$ 
depends on the model dynamics and cannot be extracted directly from option prices.
Typically, Monte Carlo simulations are used to estimate this expectation (see, e.g.,
Cao and Henry-Labordère~\cite{CL} and Lucic~\cite{L}), while 
the term $ \E^{\T}\left(x_T (x_T-k)_+ \right)$ can be evaluated directly from the implied distribution. 
Our approach consists of analytically approximating the entire combination
\begin{equation}
\label{A}
\A=\E^{\T}\left(x_T (x_T-k)_+ \right) - \E^{\T}\left(y_T\,\theta(x_T-k) \right).
\end{equation}
Note that  the quantity in the preceding display
is zero in the case when the volatility depends only on time, i.e. $\sigma(t,x) = \sigma(t)$. 
Developing the aforementioned analytical approximation is a key step 
 in obtaining  the  local volatility approximation~\eqref{main-1-factor}.
}
\end{remark}

\begin{remark}
{\rm 
As we noted above, the equation~\eqref{genericLV}  for local volatility
can be derived from some known results. 
For example,  it can be derived from ~\cite[equation (3)]{CL}.
We would 
like to note  the following concerning 
their analysis of the local volatility. 
 Despite establishing their  equation (3)
(as part of~\cite[Corollary 2.4]{CL}), they use the  equation 
\begin{align}
\label{CL-equation3}
\sigma_{loc}^2(T,k)&=2\frac{\partial _T C(T,k)+k C(T,k)+
2\int _k^{\infty }{C(T,x) dx}}{\partial_{kk}C(T,k)},
\end{align}
for the local volatility. 
  However, the above equation provides only a partial representation of the local 
 volatility (compare with~\eqref{genericLV}). 
 In particular, it includes the term 
 $$
k C(T,k)+2  \int _k^{\infty }{C(T,x) dx}= \E^\T(x_T\left(x_T-k\right)_+),
$$
that can be computed from the implied distribution,
but it omits  the term $\E^{\T}\left(y_T\,\theta(x_T-k) \right)$,
which must also  be included in the complete expression for local volatility. The missing component 
is addressed in detail below.
}
\end{remark}

\subsection{First-Order approximation}
\label{1st-order}

In this section, 
we derive a  first-order (to be explained) approximation for the implicit local volatility
 formula.

\begin{theorem}[First-Order Approximation]
\label{thm:first_order}

Consider an underlying asset whose price 
  follows a diffusion process 
 $(x_t,\, t\geq0)$ governed by  the equation
\begin{equation}
\label{eps-x}
dx_t= \sigma(t,x_t)\,dW_t, \quad x_0=0,
\end{equation}
where 
$\sigma^2(t,x)=\sigma_0^2(t) + \varepsilon\Delta \sigma^2(t,x)$.
Consider a European call option on the asset and suppose 
that given a  strike $k$  the total implied Bachelier variance satisfies the equation
\begin{equation}
\label{wT2}
w(T,k)=\int_0^T\sigma_0^2(t)dt
\end{equation}
for all maturities $T>0$.
Then under standard regularity assumptions (see Remark~\ref{assumptions} below) 
\begin{equation}
\begin{split}
A&:=\E\left(x_T(x_T-k)_+\right) - 
\E\left(\theta(x_T-k)\int_0^T\sigma^2(t,x_t)\,dt\right)\\
&= 
\frac{1}{2} p(T,k)w(T,k) \partial_kw(T,k)+o(\varepsilon),\quad\text{as}\quad \eps\to 0,
\end{split}
\end{equation}
where
\[
p(T,k) = \frac{1}{\sqrt{2\pi w(T,k)}}\exp\left( -\frac{k^2}{2w(T,k)} \right).
\]
\end{theorem}

\begin{remark}[Assumptions]
\label{assumptions}
{\rm 
By the regularity assumptions in Theorem~\ref{thm:first_order}, 
we refer to conditions sufficient to guarantee the existence of a 
unique strong solution, smooth transition densities, and 
related properties for a diffusion equation (e.g., see Ikeda and Watanabe~\cite{ike},
Karatzas and Shreve~\cite{KaratzasShreve} and
 references therein). Such assumptions are standard in the context 
 of local and stochastic volatility models in finance (see, e.g., Henry-Labordère~\cite{HenryLabordere} and
  Lorig et.al.~\cite{LorigPascucci}).
}
\end{remark}

Theorem~\ref{thm:first_order} is proved in Section~\ref{proof:first_order}. Now we apply the theorem 
to approximate the  quantity~\eqref{A}.  
To this end, note first that 
  by Gyöngy's lemma, the 
Markovian projection  $\hat{f}_t(T)$ of $f_t(T)$ 
   satisfies the equation
   \begin{equation*}
  d \hat{f}_t(T) =  \hat{\sigma}_{T}(t,\hat{f}_t(T))\,d\hat{W}_t,
\end{equation*}
where $\hat{W}_t$ is a one-dimensional standard Brownian motion and 
$$
\hat{\sigma}^2_{T}(t,k) =\E^{\T}\left(   \sigma^2_{T}(t,x_t)  \mid f_t(T) = k\right)
$$
and  $\sigma_T$ is defined in~\eqref{sigma_T}.
By properties of the  Markovian projection the marginal distributions of 
$f_t(T)$ and  $\hat{f}_t(T)$ are identical for any $t\in[0,T]$. 
Therefore, 
\begin{equation}
\label{a1}
\mathbb{E}^{\T}\left(x_T\left(x_T-k\right)_+\right) = 
\mathbb{E}^\T\left(f_T(T)(f_T(T)-k)_+\right)=
\mathbb{E}^\T\left(\hat{f}_T(T)(\hat{f}_T(T)-k)_+\right).
\end{equation}
Further, express $y_T$  explicitly in the integral form, i.e. 
$
  y_T
  =  \int_{0}^{T}  \sigma^2_{T}(t,x_t)\,dt,
$
and approximate
\begin{equation}
\label{a2}
 \mathbb{E}^{\T}\left(\theta(x_T-k)y_T\right) = 
 \int_0^T \mathbb{E}^{\T}\left(\theta(x_T-k)\sigma^2_T(t,x_t)\right)dt
 \approx  \int_0^T   \mathbb{E}^\T\left(\theta(\hat{f}_T(T)-k)\, 
  \hat{\sigma}^2_{T}(t,\hat{f}_t(T))\right)dt.
\end{equation}
Combining~\eqref{a1} with~\eqref{a2} and applying 
Theorem~\ref{thm:first_order} to the process $\hat{f}_t(T)$ gives the following approximation
\begin{equation}
\label{eq:firstOrder}
\A=\mathbb{E}^\T\left(x_T\left(x_T-k\right)_+\right)
  - \mathbb{E}^\T\left(y_T\theta(x_T-k)\right) 
\approx \frac{1}{2}p(T,k)w(T,k)\partial_kw(T,k).
\end{equation}
Denote for short $w=w(T,k)$ and $\partial_kw=\partial_kw(T,k)$.
Using~\eqref{eq:firstOrder},
the standard (in the Bachelier setting) equations
\begin{align}
\frac{\partial_{kk}C(T,k)}{p(T,k)}&=
\left(1 - \frac{k\,\partial_k w}{2w} \right)^2 +
 \frac{1}{2}\left( \partial_{kk}w - \frac{(\partial_k w)^2}{2w} \right),
 \nonumber\\
 \label{eq:CminusTerm}
C(T,k)-k\partial_kC(T,k)
&=p(T,k)\left( w - \frac{1}{2}k\partial_kw\right),\\
\nonumber
\partial_TC(T,k)&=p(T,k)\partial_Tw,
\end{align}
and the implicit formula~\eqref{genericLV}
 gives  the following approximation for the local volatility
\begin{equation}
\label{eq:first_order_LV}
\sigma^2(T,k) \approx 
\frac{ \partial_T w + \mu\left( 2w - k\partial_k w \right) + w\partial_k w }
{\left(1 - \frac{k\,\partial_k w}{2w} \right)^2 +
 \frac{1}{2}\left( \partial_{kk}w - \frac{(\partial_k w)^2}{2w} \right)}.
\end{equation}
We refer to~\eqref{eq:first_order_LV} as the \emph{first-order approximation} of the local volatility, 
as it is 
based on the  first-order terms of the perturbation expansion.

\subsection{Third-Order adjustment}
\label{3rd-order}
 
 Numerical tests indicate that the approximation~\eqref{eq:first_order_LV}
  tends to underestimate option values for long maturities.
A natural remedy is to consider a higher-order perturbation. However, this direct approach leads 
to rather cumbersome computations.
We instead propose an effective refinement of the approximation~\eqref{eq:first_order_LV} 
by incorporating a higher-order correction term.
This term arises naturally under the assumption that the implied variance is 
linear in the strike for a fixed maturity.
We then extend this idea to construct a similar approximation in the general case (see below).

Specifically,  fix maturity $T$
and  approximate the corresponding implied Bachelier 
variance by a linear function of the strike
as follows
 \begin{equation}
 \label{eq:affineW}
w(T,k)=a+bk
\end{equation}
for some $a>0$ and $b\neq 0$.
  This simple linear approximation of the implied variance turns out to be consistent with the
   approximation~\eqref{eq:firstOrder} for the local volatility.
In particular, note that this approximation is equivalent to replacing the 
conditional total variance by its implied counterpart, that is,
  \begin{equation}
\label{eq:condVarianceApprox}
\E\left(\int_{0}^{T}\sigma^{2}(t,x_t)\,\dd t\;\Bigm|\;x_T=x\right)
\;\approx\;w(T,x),
\end{equation}
so that 
\begin{equation}
\label{eq:Aheuristic}
\A\approx
\E\left(x_T(x_T-k)_+\right)-
\E\left(w(T,x_T)\,\theta(x_T-k)\right).
\end{equation}
Then, assuming~\eqref{eq:affineW} and 
expanding the {\it right-hand side} of~\eqref{eq:Aheuristic} in powers of~$b$ reproduces
\eqref{eq:firstOrder} (we omit details).

Furthermore, 
note that the use of~\eqref{eq:condVarianceApprox} (or, equivalently,~\eqref{eq:firstOrder}) has the 
following drawback.
Namely, since $\E(x_t)=0$, we have that 
\[
\E\left(x_T(x_T-k)_+\right)\to \E(x_T^{2})=
\E\left(\int_{0}^{T}\sigma^{2}(t,x_t)\,\dd t\right),\quad\text{as}\quad k\to -\infty.
\]
However, $\E(x_T^{2})-\E(w(T,x_T))\neq0$.
This discrepancy can be removed by 
using the adjusted approximation
\begin{equation}
\label{A-adj}
\A\approx
\E\left(x_T(x_T-k)_+\right)
-\E\left(\epsilon+w(T,x_T))\,\theta(x_T-k)\right),
\end{equation}
where $\epsilon:=\E(x_T^{2})-\E(w(T,x_T))$.
  It turns out that in the linear case, the correction term $\epsilon$ can be 
 computed analytically, as stated in Theorem~\ref{thm:linear} below.
  Moreover, the adjusted approximation in the linear case suggests an 
  effective improvement of the approximation~\eqref{eq:firstOrder} in the general setting.

 \begin{theorem}[Third-Order Correction]
\label{thm:linear}
Suppose that~\eqref{eq:affineW} holds.  Then  
$
\epsilon=\frac{1}{2}b^2
$
and
\begin{equation}
\label{A3}
\frac{\A}{p(T,k)} = \frac{1}{2}(a+bk)b+\frac{1}{2}b^3k+ o(b^3), \quad\text{as}\quad b\to 0.
\end{equation}
\end{theorem}

Theorem~\ref{thm:linear} is proved in Section~\ref{proof:linear}.
Now we apply this result for deriving the main result of the paper.
To this end, observe that 
in the linear case 
  $w=w(T,k)=a+bk$ and  $\partial_kw=\partial_kw(T,k)=b$.
  In these terms we have that 
   \begin{equation}
 \label{y_T-adj}
  \epsilon=\frac{1}{2}(\partial_kw(T,x_T))^2,\quad
   y_T\approx w(T, x_T)+\frac{1}{2}(\partial_kw(T,x_T))^2,
   \end{equation}
   and equation~\eqref{A3} becomes as follows
   \begin{equation}
  \label{A4}
\frac{\A}{p(T,k)}=
\frac{1}{2}w\partial_kw
+\frac{1}{2}(\partial_kw)^3
+o\left((\partial_kw)^3\right).
\end{equation}
The idea underlying the final formula~\eqref{main-1-factor} is to 
apply the equation~\eqref{A4}, rather than~\eqref{eq:firstOrder}, in the general case as well.
In other words, we propose to use the approximation 
\begin{equation}
\label{f1}
\sigma^2(T,k) \approx \frac{ \partial_T w + 
 \mu\left( 2w - k\partial_k w \right) +w\partial_k w + (\partial_k w)^3 }
{ \left( 1 - \frac{k\partial_k w}{2w} \right)^2 + 
\frac{1}{2}\left( \partial_{kk}w - \frac{(\partial_k w)^2}{2w} \right) }
\end{equation}
in the general setting.
It is  left to note that  
\begin{align*}
\frac{(\partial_kw)^3 }{\left( 1 - \frac{k\partial_kw}{2 w} \right)^2 + 
\frac{1}{2} \left(\partial_{kk}w - \frac{(\partial_kw)^2}{2 w} \right) } = 
(\partial_kw)^3 + o\left((\partial_kw)^3 \right) + 
O\left((\partial_kw)^3\partial_{kk}w\right),
\end{align*}
which allows to rewrite 
equation~\eqref{f1} in the final form~\eqref{main-1-factor}
(that differs from~\eqref{f1} by 
higher-order terms in \(\partial_kw\)).

Numerical experiments 
demonstrate that the adjusted approximation~\eqref{main-1-factor} is
 more effective than the first order approximation~\eqref{eq:first_order_LV}. 
For example, Figure~\ref{fig:IV} presents
 the implied volatility curve for a 10-year option, together with Monte Carlo estimates using
  (i) the first-order local volatility approximation and (ii) the third-order adjusted approximation.
\begin{figure}
\centering
\includegraphics[scale=0.8]{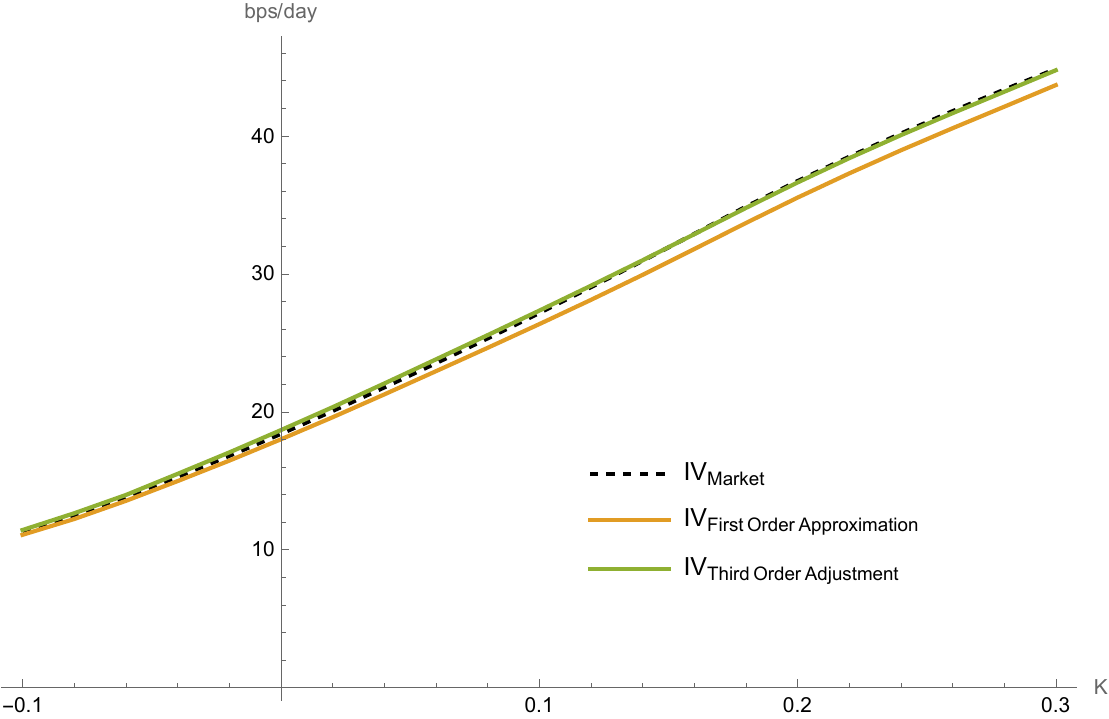}
\caption{Ten-year implied volatility curve and adjusted approximations.}
\label{fig:IV}
\end{figure}

\section{Proofs}
\label{sec:Proofs}

\subsection{Proof of Theorem~\ref{ImplicitFormula}}
\label{proofImplicit}

Fix $T > 0$ and consider a European call option with strike $k$ and 
maturity $t \leq T$ on the forward rate $f_t(T)$.  
Let 
$
C^T(t,k) := \mathbb{E}^{\T}\left(\left( f_t(T) - k \right)_+\right)
$
be the non-discounted price of the option at time $t$ under the $T$-forward measure.
At $t = T$, we have $f_T(T)=x_T$, so that 
$C^T(T,k) = \mathbb{E}^{\T}\left( x_T - k \right)_+ = C(T,k)$,
where $C(T,k)$ denotes the non-discounted price of the option on the rolling forward rate.
According to the Dupire's framework, 
we have that 
\begin{equation}
\label{Dupire1}
 \sigma^2(T,k)  =  2\,\frac{ \partial_t C^T(t,k)\big|_{t=T} }{ \partial_{kk} C^{T}(T,k) }
 =2\,\frac{ \partial_t C^T(t,k)\big|_{t=T} }{ \partial_{kk} C(T,k) }.
\end{equation}
The denominator 
$\partial_{kk} C^{T}(T,k)=\partial_{kk} C(T,k)$ can be extracted from the market, but the numerator 
$\partial_t C^T(t,k)\big|_{t=T}$ is not directly observable, since options with maturities $t<T$ are not traded.
We therefore relate $\partial_t C^T(t,k)\big|_{t=T}$ to market data.

First, note  that 
\begin{equation}
\label{total_diff}
\partial_t C^T(t,k)\big|_{t=T} = d_T C^T(T,k) - \partial_T C^T(t,k)\big|_{t=T},
\end{equation}
where $d_T$ denotes the full derivative with respect to $T$,  
and 
\[
d_T C^T(T,k) = \partial_T C(T,k).
\]
Therefore, it suffices to compute the derivative $\partial_T C^T(t,k)\big|_{t=T}$.
To this end, recall that
\[
C^{T+dT}(t,k) = \mathbb{E}^{\T+dT}\left( f_t(T+dT) - k \right)_+ 
= \mathbb{E}^{\T}\left( P_t(T+dT)\, (f_t(T+dT) - k)_+ \right)
\]
and compute 
\begin{equation*}
\begin{split}
-\partial_T C^T(t,k)\big|_{t=T}
&= -\mathbb{E}^{\T}\left( \partial_T \left( P_t(T)(f_t(T) - k)_+ \right) \Big|_{t=T} \right) \\
&= - \mathbb{E}^{\T}\left( (f_T(T)-k)_+\,\partial_T P_t(T)\big|_{t=T} \right) - 
\mathbb{E}^{\T}\left( P_t(T)\,\partial_T (f_t(T)-k)_+\big|_{t=T} \right) .
\end{split}
\end{equation*}
Since $P_T(T) = 1$, the second term simplifies, and we have that
\begin{align}
\label{eq:generic}
-\partial_T C^T(t,k)\big|_{t=T}
&= -\mathbb{E}^{\T}\left( (f_T(T)-k)_+\,\partial_T P_t(T)\big|_{t=T} \right) - 
\mathbb{E}^{\T}\left( \theta(f_T(T)-k)\,\partial_T f_t(T)\big|_{t=T} \right).
\end{align}
By~\eqref{f(t,T)-0}), 
$
\partial_T P_t(T)\big|_{t=T} = -f_T(T) = -x_T$ and $\partial_T f_t(T)\big|_{t=T} = y_T - \mu x_T$.
Therefore, 
\[
-\partial_T C^T(t,k)\big|_{t=T}
= \mathbb{E}^{\T}\left( x_T(x_T-k)_+ \right) -
 \mathbb{E}^{\T}\left( y_T\,\theta(x_T-k) \right) + \mu\,\mathbb{E}^{\T}\left( x_T\,\theta(x_T-k) \right).
\]
Using 
$
x_T\,\theta(x_T-k) = (x_T-k)_+ - k\,\partial_k (x_T-k)_+
$
we obtain that 
\begin{equation}
\label{partial_T-C}
\begin{split}
-\partial_T C^T(t,k)\big|_{t=T}
&= \mu\left( C(T,k) - k\,\partial_k C(T,k) \right) +
 \E^{\T}\left( x_T(x_T-k)_+ \right) - \E^{\T}\left( y_T\,\theta(x_T-k) \right).
\end{split}
\end{equation}
Substituting this result into~\eqref{total_diff} and recalling~\eqref{Dupire1} gives 
that
\begin{align*}
\sigma^2(T,k) &= 2\frac{ \partial_T C(T,k) - \mathbb{E}^{\T}(x_T-k)_+
\partial_T P_t(T)\big|_{t=T} - 
\mathbb{E}^{\T} \theta(x_T-k)\,\partial_T f_t(T)\big|_{t=T} }
{ \partial_{kk} C(T,k) }\\
&= 2\,\frac{ \partial_T C(T,k) + \mu\left( C(T,k) - k\,\partial_k C(T,k) \right) + 
\mathbb{E}^{\T}\left( x_T(x_T-k)_+ \right) - \mathbb{E}^{\T}\left( y_T\,\theta(x_T-k) \right) }
{\partial_{kk} C(T,k)},
\end{align*}
which proves Theorem~\ref{ImplicitFormula}.

\subsection{Proof of Theorem~\ref{thm:first_order}}
\label{proof:first_order}

Fix $T>0$ and consider
 the semigroup $(P_{t}^{\varepsilon},\, t\in [0, T])$ corresponding to the diffusion process $(x_t,\, t\in[0,T])$
 that follows equation~\eqref{eps-x}, 
i.e.
$$P_{t}^{\varepsilon}q(x)=\E(q(x_t)|x_0=x)$$
for an appropriate function $q$.
For ease of notations, we will write $P_{t}^{\varepsilon}q=\E(q(x_t))$.
In addition, denote
$$\xi(x)=x(x-k)_+,\quad \eta(x)=(x-k)_+,\quad \theta_k(x)=\theta(x-k)
\quad\text{and}\quad  \delta_k(x)=\delta(x-k),$$
where $\delta$ is the Dirac delta-function, 
and recall some useful equations
\begin{equation}
\label{used-equations}
\partial_{xx}\xi(x)=2\theta_k(x)+x\delta(x-k),\quad \partial_{xx}\eta(x)=\delta(x-k).
\end{equation}
In the above notations we have that
\begin{align*}
\E(x_T(x_T-k)_+)&=P_T^\eps\xi,\\
\E\left(y_T\theta(x_T-k)\right)&=\E\left(y_T\theta_k(x_T)\right)
=\int_0^TP_t\sigma^2(x_t,t)P^\eps_{T-t}\theta_kdt
=(P^{\varepsilon}\ast \sigma^2 P^{\varepsilon})_{T} \theta_k, 
\end{align*}
where $\ast$ denotes the convolution,  and 
by $\sigma^2$ we denot the operator of multiplication on the function $\sigma^2(x_t,t)$, so that 
 \begin{equation}
 \label{A5}
A 
=P^{\varepsilon}_{T}\xi - 
(P^{\varepsilon} \ast\sigma^2P^{\varepsilon})_{T}\theta_k.
\end{equation}
Observe now that 
the generator of the semigroup is 
$$G^{\eps}=\frac{1}{2} 
\sigma^2(t,x)\partial_{xx}=\frac{1}{2}\left(\sigma_0^2(t)+\eps\Delta\sigma^2(t,x)\right)\partial_{xx}
=G+\eps B,
$$
where 
$B = \frac{1}{2} \Delta \sigma^2(t,x)\partial_{xx}$ 
and the operator 
$
G = \frac{1}{2} \sigma_0^2(t)\partial_{xx}
$
 is the generator of the semigroup 
$(P^0_{t},\, t\in [0,T])$ of the
 diffusion  process  with the zero drift and 
time-dependent local volatility $\sigma_0(t)$, i.e. the process corresponding 
to the case when $\eps=0$.
Thus, 
the semigroup  $P_{t}^{\varepsilon}$  can be regarded as 
 a perturbation of the (unperturbed) semigroup $P^0_{t}$, which is determined by the equation
 \[
P_{t}^0 = e^{\frac{1}{2} w_{t,T} \partial_{xx}},\quad t\in[0,T],
\]
where
\begin{equation}
\label{w:eps=0}
w_{t,T} = \int_t^T \sigma_0^2(s)\,ds,\quad t\in [0,T].
\end{equation}
By Duhamel’s principle, 
\begin{equation}
\label{Duh}
\begin{split}
P_{T}^{\varepsilon} &= P_{T}^0 + \varepsilon \int_0^T P_{t}^{\varepsilon} B P_{T-t}^0 \,dt
=P_{T}^0 + \varepsilon \left(P^{\varepsilon} \ast B P^0\right)_{T}
 = P_T^0 + \frac{1}{2} \varepsilon (P^{\varepsilon} \ast \Delta \sigma^2 \partial_{xx}P^0)_T,
\end{split}
\end{equation}
so that the equation~\eqref{A5} becomes
\[
A = P_T^0\xi + \frac{1}{2} \varepsilon 
(P^{\varepsilon} \ast \Delta \sigma^2 \partial_{xx} P^0)_{T}\xi-
 (P^{\varepsilon} \ast \sigma^2 P^{\varepsilon})_{T}\theta_k.
\]
A direct calculation gives that
\begin{equation}
\label{P-xi=w}
P_{T}^0\xi= w_{0,T} P^0_{T}\theta_k,
\end{equation}
where $w_{0,T}$ is defined in~\eqref{w:eps=0}.
Using commutativity of operators $P^0_T$ and $\partial_{xx}$, and~\eqref{used-equations}
 we obtain that 
\begin{equation}
\label{partial-P-xi}
\partial_{xx} P_{T}^0\xi = 
 P_{T}^0 \partial_{xx} \xi= 
 P_{T}^0\left(2\theta_k + x\delta_k \right),
\end{equation}
Furthermore, observe that
\begin{equation}
\label{=0}
(P^{\varepsilon} \ast \Delta\sigma^2P^0)_T\delta_k = 0.
\end{equation}
Indeed, by the assumption~\eqref{wT2}, we have 
the pricing equality
$
C(T,k) = P^{\varepsilon}_T\eta = P^0_T\eta,
$
where $C(T,k)$ denotes the non-discounted 
price under the  perturbed dynamics.
Therefore, by~\eqref{Duh} and commutativity of $P^0_T$ and $\partial_{xx}$, we obtain that 
\[
 (P^{\varepsilon}\ast\Delta\sigma^2
\partial_{xx}P^0)_T\eta=(P^{\varepsilon}\ast\Delta\sigma^2P^0)_T
\partial_{xx}\eta=0.
\]
Since $\partial_{xx}\eta=\delta_k$ (see~\eqref{used-equations}, we get~\eqref{=0},
as claimed.

Further, by~\eqref{P-xi=w},~\eqref{partial-P-xi} and~\eqref{=0},
\begin{equation*}
A=w_{0,T}P_T^0\theta_k+
\varepsilon (P^{\varepsilon}\ast\Delta\sigma^2P^0)_T\theta_k-(P^{\varepsilon}
\ast\sigma^2P^{\varepsilon})_T\theta_k.
\end{equation*}
Recalling that $\sigma^2=\sigma_0^2+\Delta\sigma^2$ and using 
the equation
$(P^{\varepsilon}\ast\sigma^2_0P^{\varepsilon})_T=w_{0,T}P_T^{\varepsilon}$
 rewrite the preceding equation for $A$
as follows
\begin{align}
\label{A6}
A
&=w_{0,T}\Big(P^0_T-P^{\varepsilon}_T\Big)\theta_k-
\varepsilon\,
(P^{\varepsilon}\ast\Delta\sigma^2(P^0-P^\eps))_T\theta_k.
\end{align}
By assumption~\eqref{wT2}, 
$w_{0,T}=w(T,k).$
Regarding the Bachelier's price as a function of the strike and the
 total variance $w$ (see Remark~\ref{BH(w)}) we have the following 
 equations
\begin{align*}
P^0_T\theta_k&=-\partial_k\mathrm{BH}(k,w(T, k)),\nonumber\\[1mm]
P^{\varepsilon}_T\theta&
=-\partial_k\mathrm{BH}(k,w(T, k))-\partial_w\mathrm{BH}(k,w(T, k))\partial_kw(T,k)\\
&=-\partial_k\mathrm{BH}(k,w(T, k))-
\frac{1}{2}p(T,k)\partial_kw(T,k),
\end{align*}
where 
\[
p(T,k)=\partial_w\mathrm{BH}(k,w(T,k))=\frac{1}{\sqrt{2\pi w(T,k)}} \exp\left( -\frac{k^2}{2 w(T,k)} \right).
\]
Collecting the above equations 
gives  the following for  the first term in the right-hand side of~\eqref{A6}  
$$
w_{0,T}\Big(P_T^0-P^{\varepsilon}_T\Big)\theta_k
=\frac{1}{2}p(T,k) w(T,k) \partial_kw(T,k).
$$
Finally, since the  term 
$- \varepsilon (P^{\varepsilon} \ast 
 \Delta \sigma^2(P^0-P^\eps))_T\theta_k$ in~\eqref{A6} is of order $\eps^2$,
  we conclude that
\[
A =  \frac{1}{2} p(T,k)w(T,k) \partial_kw(T,k) + o(\varepsilon),
\]
which proves Theorem~\ref{thm:first_order}.

\subsection{Proof of Theorem~\ref{thm:linear}}
\label{proof:linear}

Start with the following lemma, which might be of interest on its own right.
\begin{lemma}
\label{IG}
$$
\E(x_T(x_T-k)_+)
=
\frac{1}{2}\mathbb{E}\left(\tau\,{\bf 1}_{\{\tau>a+bk\}} \right)
+\frac{1}{2a}\mathbb{E}\left(\tau^2\,{\bf 1}_{\{\tau>a+bk\}} \right),
$$
where $\tau$ is a random variable which has an inverse Gaussian distribution  (IG distribution, 
aka Wald distribution, e.g., see Johnson et.al.~\cite{JohnsonKotz} and references therein)
with parameters $a$ and $a^2/b^2$, i.e. $\tau\sim \mathrm{IG}(a,a^2/b^2)$.
\end{lemma}

\begin{proof}[Proof of Lemma~\ref{IG}]

Start with expressing the expectation  $\E(x_T(x_T-k)_+)$
 in terms of the call value $C(T,k)=\E((x_T-k)_+)$ (with maturity $T$ and strike $k$).
To this end, recalling  that the second derivative 
 $\partial_{xx} C(T,x)$ is equal to the pdf of $x_T$,
and integrating by parts, we get that 
\begin{align*}
\E(x_T(x_T-k)_+)
&
=\int_k^{\infty}
x(x-k)\partial_{xx}C(T,x)\,dx\\
&=\left(x(x-k)\partial_xC(x))\right]_k^{\infty}-\int_k^{\infty}(2x-k)\partial_x C(x)\,dx\\
&=-2\int_k^{\infty}x\partial_xC(x)\,dx-kC(k)\\
&=
k\,C(T,k)+2\int_{k}^{\infty}C(T,x)\,dx.
\end{align*}
Then, since
$$kC(T,k)+\int_k^{\infty}C(T,x)\,dx=\int_k^{\infty}x\partial_xC(x)\,dx,$$
we obtain that 
\begin{equation*}
\E(x_T(x_T-k)_+)
= \int_{k}^{\infty} \left(C(T,x)-x\partial_x C(T,x) \right)\,dx.
\end{equation*}
In the linear case $w(T,x)=a+bx$ 
the equation~\eqref{eq:CminusTerm} reduces to
\begin{equation*}
C(T,x)-x\,\partial_x C(T,x)
= p(T,x) \left( a+\frac{b}{2}x \right),
\end{equation*}
where 
$$p(T, x)=\frac{1}{\sqrt{2\pi(a+bx)}} \exp\left(-\frac{x^2}{2(a+bx)}\right),\quad x>-\frac{a}{b}.$$
Therefore,
\begin{equation*}
\E(x_T(x_T-k)_+)=\int_k^{\infty}p(T,x) \left(a+\frac{b}{2}x \right)dx
=
\frac{1}{2b} \int_{a+bk}^{\infty}
\frac{\exp\left( -\frac{(u-a)^2}{2b^2u} \right)}{\sqrt{2\pi} u^{3/2}} (a+u) u\,du.
\end{equation*}
It is left to note that  the function 
$$g(u)=\frac{a}{b}\frac{1}{\sqrt{2\pi}u^{3/2}}\exp\left(-\frac{(u-a)^2}{2b^2u}\right),\quad u>0,$$
is the density of an inverse Gaussian (IG) distribution, $\mathrm{IG}\left(a,a^2/b^2\right)$. The lemma is proved.
\end{proof}

By the properties of  IG distribution, 
$$
\E(\tau)=a \quad\text{and}\quad \E(\tau^2)=a^2+ab^2.
$$
Therefore, since  $\E(x_T)=0$, we get  that 
$$
\E(x_T^2)
= \frac{1}{2}\E(\tau)+\frac{1}{2a}\E(\tau^2)
= a + \frac{1}{2}b^2,
$$
and, hence, the adjustment term in the linear case  is 
\begin{equation*}
\epsilon
= \E(x_T^2)- \E(a+bx_T)
= \E(x_T^2)-a= \frac{1}{2}b^2,
\end{equation*}
as claimed.

Further, the   truncated 
moments of $\tau\sim \mathrm{IG}\left(a,a^2/b^2\right)$  are given by 
\begin{align*}
\E(\tau\,{\bf 1}_{\{\tau>a+bk\}})
   &=a\,\Phi\bigl(-\delta_1(a+bk)\bigr)
     +a\,e^{2a/b^2}\,\Phi\bigl(\delta_2(a+bk)\bigr), \notag\\
\E(\tau^{2}\,{\bf 1}_{\{\tau>a+bk\}})
   &=b^2\,e^{2a/b^2}
     \frac{2a}{b}\sqrt{a+bk}\phi\left(\delta_2(a+bk)\right)\\
     & +e^{-2a/b^2}\left(a+\frac{a^2}{b^2}\right)\Phi\!\left(-\delta_1(a+bk)\right)
  +\left(a-\frac{a^2}{b^2}\right)\Phi\left(\delta_2(a+bk)\bigr)\right),
\end{align*}
where $\Phi$ and $\phi$ are the cumulative distribution function and the probability density function 
of the standard normal distribution,
respectively,
\[
\delta_1(x)=\frac{a}{b\sqrt{x}}\Bigl(\frac{x}{a}-1\Bigr)
\quad\text{and}\quad
\delta_2(x)=-\frac{a}{b\sqrt{x}}\Bigl(\frac{x}{a}+1\Bigr).
\]
Expanding in a Taylor series with respect to the parameter $b$
yields that
$$
\frac{\A}{p(T,k)} = \frac{1}{2}b(a+bk+b^2) + o(b^3).
$$
The theorem is proved.

\section{Multi-factor model}
\label{SecNFactorCase}

The local volatility framework can be extended to to multi-factor Cheyette models. The extension is relatively straightforward, so 
we provide only a brief outline of the procedure, highlighting the main modifications required. 

Start with recalling the model  (see, e.g., Andersen and Piterbarg~\cite{AndersenPiterbarg2}).
 Under the risk-neutral measure \(\mathbb Q\) the model equations are
\begin{align*}
  d x_t &= \bigl( y_t \e-\mu x_t\bigr)\,dt + \sigma_r(t,x_t)\,dW_t,\\
  d y_t &= \bigl(\sigma_r(t,x_t)\sigma_r(t,x_t)^{\top}-\mu y_t-y_t\mu^{\top}\bigr)\,dt,
 \label{eq:ySDE}
\end{align*}
where \(x_t=(x_{1,t},\ldots, x_{d,t})^\top\in\mathbb R^d\) is the state vector, 
\(y_t=(y_{ij,t})_{i,j=1}^d\) is a real $d\times d$ symmetric matrix, 
$
{\bf\mu}=\mathrm{diag}(\mu_1,\dots,\mu_d)
$
is 
a  diagonal mean reversion matrix, $ \sigma_r(t,x)$ is a volatility matrix, 
 $\e\in\R^d$ is a vector whose components are all equal to $1$ and $W_t=(W_{1,t},\ldots,W_{d,t})^\top$
 is a d-dimensional standard Brownian motion.
 We assume that 
 $
 \sigma_r(t,x)=\sigma(t, \bar{x}) V,
 $
  where $ \bar{x} = \e^\top x = \sum_{i=1}^dx_{i}$ for $x=(x_1,\ldots,x_d)^\top\in \R^d$, 
  $\sigma(\cdot, \cdot)$ is a deterministic function of two
  variables and $V$ is a $d\times d$ real matrix which satisfies 
 the normalisation condition
$\e^{\top}VV^{\top}\e=1$.
In these notations,
\begin{align*}
  d x_t &= \bigl(y_t \e -\mu x_t\bigr)\,dt + \sigma(t, \bar{x}_t) V dW_t,\\
  d y_t &= \bigl(\sigma^2(t, \bar{x}_t)V V^{\top}-\mu y_t-y_t\mu^{\top}\bigr)\,dt.
\end{align*}
In the multi-factor case
the bond price and the  instantaneous forward rate are given by 
\begin{equation}
\label{P-mult}
  P_t(T)=e^{-G(T-t)^{\top}x_t-\tfrac12\,G(T-t)^{\top}y_t\,G(T-t)}
\end{equation}
and 
\begin{align}
\label{forw-mult}
 f_t(T)&=\e^{\top}e^{-\mu (T-t)}(x_t + y_t\,G(T-t)),\\
r_T &= \e^{\top}x_T 
\end{align}
respectively, 
where $G(t) = {\mu}^{-1}(1-e^{-\mu t})\e$.

\begin{theorem}
\label{ImplicitFormula2F}
In the multi-factor Cheyette model, the local volatility is given by
\begin{equation*}
\label{localVol2F}
\begin{split}
&\sigma^2(T,k) \\
&=
 2\,\frac{ \partial_T C(T,k) + \mu_{{\rm eff}}\left( C(T,k) - k\,\partial_k C(T,k) \right) + 
  \mathbb{E}^{\T}\left(r_T(r_T-k)_+\right) -
   \mathbb{E}^{\T}\left(\theta(r_T-k)\, \e^\top y_T\e\right)}
{ \partial_{kk} C(T,k) },
\end{split}
\end{equation*}
where
$$
\mu_{{\rm eff}} = 
\frac{\mathbb{E}^{\T}\left(\theta(r_T-k)\, \e^{\top} 
\mu x_T\right) }{\mathbb{E}^{\T}\left( \theta(r_T-k)\, r_T\right)}.
$$
\end{theorem}

\begin{proof}[Proof of Theorem~\ref{ImplicitFormula2F}]
\label{proof:ImplicitFormula2F}
The proof of Theorem~\ref{ImplicitFormula2F} follows the same 
steps as in the one-factor case, with additional details provided below.
First, under the $\T$‑forward measure $\mathbb Q^{T}$ we have the equation for the instantaneous 
forward rate
$$ df_t(T)= \sigma_T(t,x_t) \,dW_t^{\T},\quad, t\in [0,T],$$
where    $\sigma_T(t,x)=\e^{\top}e^{-\mu (T-t)} \sigma_r(t,x)$.
The  matrix  $y_t$ can be expressed explicitly in the integral form
\[
  y_t  =  \int_{0}^{t} e^{-\mu (t-s)}\, \sigma_r(s)\sigma_r(s)^{\top}\, e^{-\mu (t-s)}\,ds
\]
and,  hence, the  total variance is given by 
$$
\e^{\top} y_T\e =  \int_{0}^{T}  \sigma_T(t,x_t) \sigma_T(t,x_t)^{\top}\,dt.
$$
By Gyöngy's lemma, 
a one-dimensional Markovian projection  $\hat{f}_t(T)$ of $f_t(T)$
follows the equation
\begin{equation*}
  d \hat{f}_t(T) =  \hat{\sigma}_{T}(t,\hat{f}_t(T))\,d\hat{W_t},
\end{equation*}
where  $\hat{W}_t$ is a one-dimensional standard Brownian motion and 
\[
\hat{\sigma}_{T}(t,k)^2 =\E^{\T}\left(   \sigma_T(t,x_t)   \sigma_T(t,x_t)^{\top} \mid f_t(T) = k\right).
\]
The standard Dupire equation is applicable for process $\hat{f}_t(T)$ i.e.
$$
\hat{\sigma}^2_{T}(t,k) = 2\,\frac{ \partial_t C^T(t,k) }{ \partial_{kk}  C^T(t,k) }.
$$
For $t=T$ we have that
\begin{align*}
  \hat{\sigma}^2_{T}(T,k)&=
  \E^{\T}\left(  \e^{\top} \sigma_r(T,x_T)  \sigma_r(T,x_T)^{\top}\e  \mid f_T(T) = k\right) \\
&=\E^{\T}\left( \sigma^2(T,\bar{x}_T)\e^{\top} V V^{\top}\e  \mid\bar{x}_T = k\right)
= \sigma^2(T,k). 
\end{align*}
Hence,
$$
 \sigma^2(T,k) =  2\,\frac{ \partial_t C^T(t,k)\big|_{t=T} }{ \partial_{kk} C(T,k) }.
$$
The rest of the proof is similar to the one-factor case. Namely, 
it remains to use 
equation~\eqref{eq:generic}  for expressing 
the derivative $\partial_tC^T(t,k)\big|_{t=T}$ in terms of derivatives 
 $\partial_T P_t(T)\big|_{t=T}  = -\e^{\top}x_T= -f_T(T)$ (i.e. of the  bond price~\eqref{P-mult})   and
     $\partial_T f_t(T)\big|_{t=T} =\e^{\top} y_T\e - \e^{\top} \mu x_T$ (i.e. of 
      the instantaneous forward rate~\eqref{forw-mult}) and 
   to apply the analogue of equation~\eqref{partial_T-C} to complete the proof (we skip the details).
\end{proof}

By property of the 
Markovian projection marginal distribution of $f_t(T)$ and  $\hat{f}_t(T)$ at $t = T$ are identical.
Hence, we have the identity 
$$\mathbb{E}^T\left(f_T(T)\left(f_T(T)-k\right)_+\right) =
 \mathbb{E}\left(\hat{f}_T(T)(\hat{f}_T(T)-k)_+\right)$$
 and the approximation
\begin{align*}
 \mathbb{E}^{\T}\left(\theta(f_T(T)-k)\e^{\top} y_T\e\right)& = 
 \int_0^T \mathbb{E}^{\T}
 \left(\theta(f_T(T)-k)\sigma_T(t,x_t) {\sigma_T(t,x_t) }^{\top}\right)dt\\
 &
\approx \mathbb{E}\left(\theta(\hat{f}_T(T)-k)\,  \hat{\sigma}^2_{T}(t,\hat{f}_T(t))\right).
\end{align*}
Proceeding similarly to the one-factor case,  we obtain 
the following approximation for the multi-factor analogue of the quantity $\A$
\begin{align*}
&
\mathbb{E}^T\left(f_T(T)\left(f_T(T)-k\right)_+\right)  - \mathbb{E}^T\left(
 \theta\left(f_T(T)-k\right)\e^{\top} y_T\e\right) \\
&\approx \mathbb{E}\left(\hat{f}_T(T)(\hat{f}_T(T)-k)_+\right) - 
\mathbb{E}^T\left(\theta(\hat{f}_T(T)-k )   \int_0^T  \,  \hat{\sigma}^2_{T}(t,\hat{f}_T(t)) dt\right) \\
&\approx \frac{1}{2}p(T,k)\left( (\partial_kw)^3 + \partial_kw \right),
\end{align*}
which gives the approximation
$$
\sigma^2(T, k)
 \approx \frac{ \partial_T w + \mu_{\text{eff}} \left( 2w - k\,\partial_k w \right) + w\,\partial_k w }
{\left(1 - \frac{k\,\partial_k w}{2w} \right)^2 +
 \frac{1}{2}\left( \partial_{kk}w - \frac{(\partial_k w)^2}{2w} \right)} + (\partial_k w)^3
$$
for the local volatility in the multi-factor case.

\begin{remark}
{\rm 
Note that, in comparison with the one-factor case, the following additional 
term arises in the multi-factor setting
$\mathbb{E}^{\T}\left( \theta(f_T(T)-k)\e^{\top} \mu x_T   \right)$.
In the  Gaussian case this term can be estimated by
$$
 \mathbb{E}^{\T}\left( x_{T}\Big|\e^{\top} x_T = r \right) =
  r \frac{\Cov(x_{T}, \e^{\top} x_T)}{\Var(\e^{\top} x_T)}=  r \frac{\bar{y}_T\e}{\e^{\top} \bar{y}_T\e},
$$
 where 
$\bar{y}_T$ denotes the matrix $y_T$, calibrated within the 
Gaussian model to match the at-the-money (ATM) term structure of the  implied volatility.
This gives 
 the following equation for the effective mean-reversion
$$
\mu_{\text{eff}} =  \frac{\e^{\top} \mu \bar{y}_T\e}{\e^{\top} \bar{y}_T\e}.
$$
In Section~\ref{2-factor} we show how 
 this quantity can be evaluated in the case of the two-factor model. 
 }
 \end{remark}

\section{Example: the two-factor model}
\label{2-factor}

To illustrate the multi-factor extension, we now examine the two-factor Cheyette model. 
Let 
\[
x_t = 
\begin{pmatrix}
x_{1,t} \\
x_{2,t}
\end{pmatrix},
\quad
y_t =
\begin{pmatrix}
y_{1,t} & y_{3,t} \\
y_{3,t} & y_{2,t}
\end{pmatrix},
\quad
\mu = 
\begin{pmatrix}
\mu_1 & 0 \\
0 & \mu_2
\end{pmatrix},
\]
where, for convenience, we denoted $y_{11,t}=y_{1,t},\,  y_{22,t}=y_{2,t}$ and $y_{12,t}=y_{21,t}=y_{3,t}$.
The volatility matrix is given by 
\[
\sigma_r(t,x_t) = \sigma(t, x_{1,t}+x_{2,t})
\begin{pmatrix}
\alpha & 0 \\
\rho \beta & \sqrt{1-\rho^2}\beta
\end{pmatrix},
\]
where $-1 \leq \rho \leq 1$, $\alpha>0$ and $\beta>0$ are given constants, 
and the normalization condition is 
\begin{equation}
\label{norm-2}
\e^{\top}VV^{\top}\e=1 \iff
\alpha^2 + 2\rho \alpha \beta + \beta^2 = 1.
\end{equation}
In the Gaussian case, i.e. when $\sigma^2(t,x)=\sigma^2(t)$,
the following approximation 
holds for the effective mean-reversion parameter  
 (see Section~\ref{Gaussian} for details)
\[
\mu_{\text{eff}}(T) \approx \frac{\mu_1 + \mu_2}{2} + \frac{\mu_1 - \mu_2}{2w(T)}
\frac{1}{2\gamma} \int_0^T \left( e^{(t-T)\lambda_2}(a+\gamma b) - 
e^{(t-T)\lambda_1}(a-\gamma b) \right) u(t)\,dt,
\]
where 
\begin{equation}
\label{u}
u(t)= \partial_t w(t)+ (\mu_1+\mu_2) w(t),\quad t\geq 0,
\end{equation}
$w(t)=(w(t),\, t\geq 0),$ is the implied total variance,
and
\begin{align}
\label{gamma}
\gamma &= \sqrt{ (1+2\rho\alpha\beta)^2 - (2\alpha\beta)^2 },
\\
\label{lamb1}
\lambda_{1} &= \mu_1 + \mu_2 + (\beta^2-\alpha^2)
\frac{\mu_1-\mu_2}{2}+\gamma\frac{\mu_1-\mu_2}{2},
\\
\lambda_{2} &= \lambda_1- \gamma(\mu_1-\mu_2),
\label{lamb2}\\
\nonumber
a &= (\alpha^2 - \beta^2)^2 - 2(\alpha^2+\beta^2),\\
\nonumber
b &= \alpha^2 - \beta^2.
\end{align}
We use the ATM term structure of the implied variance, i.e. 
$w(t) = t\sigma^2_{\text{imp}}(t,0)$,  to evaluate $u(t)$.
Figure~\ref{fig:2F} presents the implied volatility curve for a 10-year option,
 together with Monte Carlo estimates using (i) a one-factor model 
 with the  mean reversion parameter  $\mu=0.5$ and (ii) a two-factor model
 with parameters $ \rho = 0.5$, $\mu_1 = 0.0005$, $\mu_2 = 0.5$, $\alpha = 0.7$.
\begin{figure}[h!]
\centering
\includegraphics[scale=1]{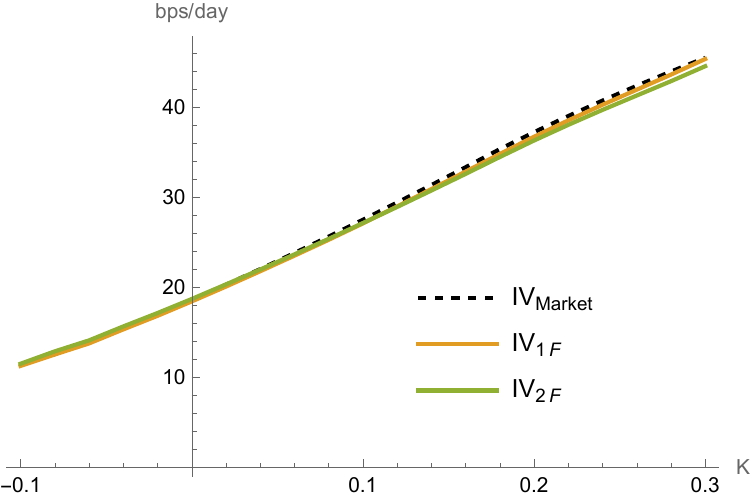}
\caption{Ten-year implied volatility curve and 1F and 2F cases.}
\label{fig:2F}
\end{figure}

\section{Calibration of the Cheyette model to swaptions}
\label{swaption:calibration}

 In this section we  discuss a method for calibrating the Cheyette 
model to swaption market. 
As mentioned above, a direct swaption-to-local-volatility
 construction is difficult in the Cheyette model because swaption payoffs depend on the entire bond curve and 
 therefore on the two-dimensional state $(x_T, y_T)$. 
In practice, however, market information is available primarily in the form of swaption prices (or implied
 volatilities). The calibration problem therefore amounts to translating swaption data into the surface $w(T,k)$. 
 In Section~\ref{calibration-framework}  we
 explain how this can be achieved by matching the density implied by swaption prices with 
 the density of the short rate.
Note that we  focus on the conceptual calibration approach 
 rather than on a specific numerical implementation.
  The density-matching construction presented here should be viewed as one possible calibration approach
  (in principle, the model could also be calibrated directly to swaption prices).
 The precise implementation may depend on data quality and practical constraints and is not pursued in detail here.

\subsection{Swaption-based calibration framework}
\label{calibration-framework}

The objective of the calibration procedure is 
 to construct  the Cheyette local-volatility surface
consistent with observed swaption market data.
We assume that the following market inputs are available on the calibration date.
\begin{itemize}
	\item \textbf{Discount curve.} A curve \( P_0(t) \) used to price fixed-leg cashflows and to define annuities.
	
	\item \textbf{Swaption smile(s).} For each swaption expiry \( T = T_0 \) and underlying swap schedule 
	\( T_0 < T_1 < \dots < T_n \), the market provides either 
	(a) payer/receiver swaption prices across strikes, or 
	(b) Bachelier  implied volatilities \( \sigma^{S}_{\mathrm{imp}}(T,K) \) across strikes \( K \), 
	quoted in absolute strikes, moneyness, or delta (converted to strikes using standard conventions).
\end{itemize}
From these quotes we construct a smooth total implied variance surface for the swap rate,
\[
z(T,K) := T\big(\sigma^{S}_{\mathrm{imp}}(T,K)\big)^2.
\]
In practice, because the local-volatility approximation (12) requires derivatives of implied variance,
it is essential that \( z(T,K) \) be represented by a smooth and arbitrage-consistent interpolant in \( K \)
(and preferably also in \( T \)).

\paragraph{Swap rate and annuity as functions of $(x_T,y_T)$.}

Consider a swap with fixing date $T = T_0$, maturity $T_n$, and payment dates $T_i$, $i=1,\dots,n$. 
The fair swap rate realized at time $T$ is
\[
S_T = \frac{1 - P_T(T_n)}{A_T},
\]
where
\[
A_T = \sum_{i=0}^{n-1} P_T(T_{i+1}) \Delta T_i
\quad\text{and}\quad
\Delta T_i = T_{i+1} - T_i,
\]
is the value at time $T$ of the fixed-leg annuity.
The time-$0$ price of a payer swaption with expiry $T$ and strike $K$ under the $T$-forward measure 
$\QT$ is given by the discounted expectation 
\begin{equation}
V(0;T,K) = P_0(T)\,\mathbb{E}^{\T}\!\left(A_T (S_T - K)^+ \right),
\label{tag57}
\end{equation}
whose evaluation requires knowledge of the joint distribution of
$(S_T, A_T)$, where 
both $A_T$ and $S_T$ are  functions of the pair $(x_T,y_T)$, 
since in the one-factor Cheyette model
\[
P_T(U;x_T,y_T) = \frac{P_0(U)}{P_0(T)}
\exp\!\left(
- G(T,U)x_T - \frac12 G^2(T,U)y_T
\right)
\quad\text{and}\quad
G(T,U) = \frac{1 - e^{-\mu(U-T)}}{\mu},
\]
for $U \ge T$.

\paragraph{One-dimensional reduction: replacing $y_T$ by a function of $x_T$.}

The central obstacle to a direct swaption-to-$w$ translation is that swaption payoffs depend on
\[
y_T = \int_0^T \sigma^2(t,x_t)\,dt,
\]
which is not determined by $x_T$ alone.
Sections~\ref {1st-order} and~\ref{3rd-order} 
motivate   an approximation that replaces $y_T$ by a function of $x_T$
expressed in terms of the short-rate implied total variance surface $w(T,k)$.
Specifically, the third-order adjustment~\eqref{y_T-adj} suggests the proxy
\[
\bar y(T,x) := w(T,x) + \frac12 \big(\partial_x w(T,x)\big)^2,
\]
interpreted as an approximation to the conditional expectation
$\mathbb{E}^{\T}(y_T \mid x_T = x)$.
This approximation is exact in the time-dependent Gaussian case
and remains accurate for moderate implied smiles.
With this approximation, we obtain deterministic representations for bond prices,
the annuity, and the swap rate as functions of $x$ alone, namely
\begin{align}
	\label{A-x}
	A(x) &:= \sum_{i=0}^{n-1}  P_T(T_{i+1};x,\bar y(T,x))\,\Delta T_i,\\
	\label{S-x}
	S(x) &:= \frac{1 -  P_T(T_n;x,\bar y(T,x))}{A(x)}.
\end{align}
The swaption payoff $A_T (S_T - K)^+$ is therefore approximated by
$
A(x_T)\big(S(x_T)-K\big)^+,
$
which is a function of just $x_T$.

\paragraph{From the swaption smile to the swap-rate density.}
Further, let
\[
C^S(T,K) := \mathbb{E}^\T\!\big(A_T (S_T-K)^+\big)
\]
be the non-discounted swaption price and recall the Breeden--Litzenberger identity
\[
C^S_{KK}(T,K) = \mathbb{E}^\T\!\big(A_T \delta(S_T-K)\big).
\]
It is convenient to measure the strike relative to the ATM swap rate and write
\[
K := K - S_0(T).
\]
Let $z(T,K)$ be the corresponding Bachelier implied total variance. Then
\begin{align*}
C^S(T,K)&=A_0\,\mathrm{BH}\!\big(K,z(T,K)\big),\\
C^S_{KK}(T,K)&=A_0\,\mathrm{BH_{KK}}\!\big(K,z(T,K)\big),
\end{align*}
and, hence,
\begin{equation}
\mathbb{E}^\T\!\big(A_T\delta(S_T-K)\big)
=
A_0\,p^{\mathrm{mkt}}_S(T,K;z),
\label{swaption-density}
\end{equation}
where
\[
p^{\mathrm{mkt}}_S(T,K;z)
= \mathrm{BH_{KK}}\!\big(K,z(T,K)\big)
=
\frac{1}{\sqrt{2\pi z}}
\exp\!\left(-\frac{K^2}{2z}\right)
\left[
\left(1-\frac{K z_K}{2z}\right)^2
+\frac12\left(
z_{KK}-\frac{z_K^2}{2z}
\right)
\right]
\]
and $z=z(T,K)$.

\paragraph{Relating the market-implied density $p^{\mathrm{mkt}}_S(T,K)$ to the density of $x_T$.}
By construction, the short-rate implied total variance surface $w(T,k)$ is defined through the
non-discounted short-rate call prices under the $T$-forward measure $\mathbb{Q}_{\T}$, i.e.
\[
C(T,k) := \mathbb{E}^\T\!\big((x_T-k)^+\big)
= \mathrm{BH}\big(k,w(T,k)\big).
\]

The density of $x_T$ under $\mathbb Q_T$ is therefore
\[
p_x(T,k)=C_{kk}(T,k)=\mathrm{BH}_{kk}(k,w(T,k)),
\]
which admits the explicit representation
\[
p_x(T,k;w)=
\frac{1}{\sqrt{2\pi w}}
\exp\!\left(-\frac{k^2}{2w}\right)
\left[
\left(1-\frac{k w_k}{2w}\right)^2
+\frac12\left(w_{kk}-\frac{w_k^2}{2w}\right)
\right].
\]

On the other hand, from~\eqref{swaption-density} the market density 
$p^{\mathrm{mkt}}_S(T,K)$ is obtained from the swaption implied variance surface $z(T,K)$.
Using the relation between $S_T$ and the state variable $x_T$, we obtain
\[
\mathbb{E}^\T\!\big(A_T\delta(S_T-K)\big)
\approx
\frac{A(x_K)}{|S'(x_K)|}\,p_x(T,x_K),
\]
where $x_K$ solves
$
S(x_K)=K.
$
Equating these two representations of $\mathbb{E}^\T\!\big(A_T\delta(S_T-K)\big)$ yields
\[
A_0\,p^{\mathrm{mkt}}_S(T,K;z)
=
\frac{A(x_K)}{|S'(x_K)|}\,p_x(T,x_K;w),
\]
which provides the calibration condition for the short-rate variance surface $w(T,\cdot)$.

\subsection{Numerical example}
\label{swap-example}

We illustrate the calibration procedure for the Cheyette model with local volatility using a numerical example. Specifically, we apply the method to a typical 5Y/5Y payer swaption and compare the resulting model-implied volatilities with market data.

The market input is given by the swaption implied volatility surface 
$\mathrm{IV}_S=\mathrm{IV}_S(k,T)$.
First, we  transform the swaption implied volatilities into implied volatilities 
$\mathrm{IV}_f=\mathrm{IV}_f(k,T)$
of rolling maturity options on the short rate.
This transformation is performed by using the calibration method 
described in the previous section.
The resulting implied forward volatilities $\mathrm{IV}_f$ are shown in Figure~\ref{fig:ivFS}.

\begin{figure}[h!]
\centering
\includegraphics[width=0.65\textwidth]{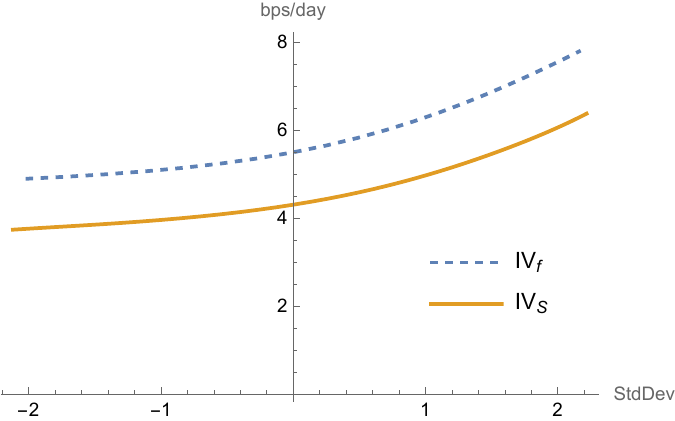}
\caption{Implied forward volatilities $\mathrm{IV}_f$ derived from swaption volatilities $\mathrm{IV}_S$.}
\label{fig:ivFS}
\end{figure}
Next, we use the obtained forward implied volatilities to calibrate the local volatility surface 
via the approximation~\eqref{main-1-factor}.
Based on  this local volatility surface, 
we simulate the dynamics of the forward rate using a Monte Carlo method.
The swaption price is then reconstructed from the simulated  paths, and
 the corresponding implied volatility is computed.
The comparison between the model-implied swaption volatilities and the original market 
swaption volatilities is shown in Figure~\ref{fig:iv}.
\begin{figure}[h!]
\centering
\includegraphics[width=0.65\textwidth]{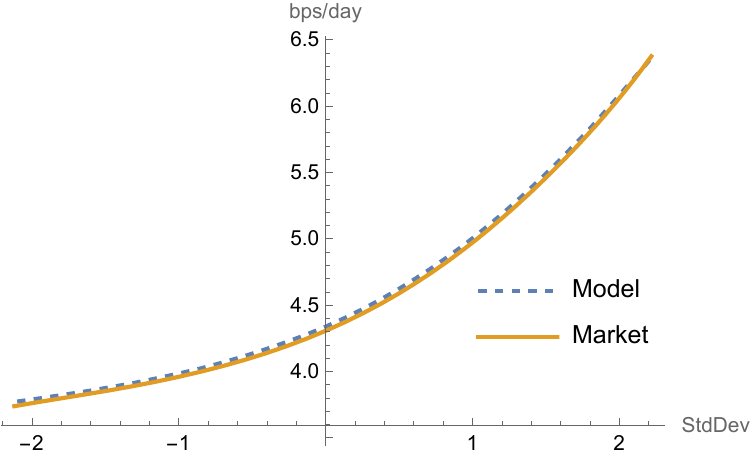}
\caption{Model-implied swaption volatilities compared to market swaption volatilities.}
\label{fig:iv}
\end{figure}
The results demonstrate a good fit between the model and the market data,
 confirming the practical applicability and accuracy of the proposed calibration method.

\section{Appendix. Two-factor  Gaussian case}
\label{Gaussian}

We derive explicit formulas for $y_{1,t}$, $y_{2,t}$, and $y_{3,t}$ in the Gaussian case, where $\sigma(t, x) = \sigma(t)$ is independent of the state variable.
In this case we have that 
\begin{equation*}
\mu_{\text{eff}}(t) =  \frac{\e^{\top} \mu y_t \e}{\e^{\top} y_t \e}= 
\frac{\mu_1 y_{1,t} + (\mu_1 + \mu_2)y_{3,t}+ \mu_2 y_{2,t}}{w(t)},
\end{equation*}
where $w(t)$ is the implied  total variance.
  Note that 
\begin{align*}
w(t) &=\e^{\top} y_t \e =  y_{1,t} +2 y_{3,t}+y_{2,t}\\
y_{3,t} &= \frac{w(t) - y_{1,t} - y_{2,t}}{2}.
\end{align*}
Therefore,
\begin{equation*}
\mu_{\text{eff}}(t)  =  \frac{\mu_1 + \mu_2}{2} + \frac{\mu_1 - \mu_2}{2} \frac{y_{1,t} - y_{2,t}}{w(t)}.
\end{equation*}

Let us  express  the variables $y_{1,t}$, $y_{2,t}$ and $y_{3,t}$  in terms of 
 $w(t)$. This will, in turn, allow us to 
 express  $\mu_{\text{eff}}(t)$ as a function of $w(t)$.
In the Gaussian case we have the  following ODEs
\begin{align*}
\partial_t y_{1,t} + 2\mu_1 y_{1,t} &= \alpha^2 \sigma^2(t), \\
\partial_t y_{2,t} + 2\mu_2 y_{2,t} &= \beta^2 \sigma^2(t), \\
\partial_t y_{3,t} + (\mu_1+\mu_2) y_{3,t} &= \rho \alpha \beta \sigma^2(t),
\end{align*}
that yield the following two  equations 
\begin{align}
\label{y11}
u(t)-\left(\mu_1+\mu_2+4\rho
 \frac{\beta }{\alpha}\mu_1\right)y_{1,t}-\left(\mu_1+\mu _2\right)y_{2,t}
&=\left(1+2\rho
 \frac{\beta }{\alpha }\right)\partial_ty_{1,t}+\partial _ty_{2,t}\\
 \label{y21}
u(t)-\left(\mu_1+\mu_2+4\rho
\frac{\alpha }{\beta } \mu_2\right)y_{2,t}-\left(\mu_1+\mu_2\right)y_{1,t}
&=\left(1+2\rho\frac{\alpha}{\beta}\right)\partial_ty_{2,t}+\partial_ty_{1,t},
\end{align}
where  $u(t)$ is the function defined in~\eqref{u}.
Introducing  the vector 
$
Y(t) = 
\begin{pmatrix}
y_{1,t} \\
y_{2,t}
\end{pmatrix}
$
  rewrite equations~\eqref{y11} and~\eqref{y21} in the 
matrix form 
\begin{align*}
u(t)\e-Q Y(t)=P\frac{dY(t)}{dt},
\end{align*}
or, equivalently, 
\begin{equation}
\label{Y}
\frac{dY(t)}{dt} = (u(t) P^{-1}\e - M Y(t)),
\end{equation}
where
\begin{align*}
P &= 
\begin{pmatrix}
1+2\rho \beta/\alpha & 1 \\
1 & 1+2\rho \alpha/\beta
\end{pmatrix},
\\
Q &= 
\begin{pmatrix}
\mu_1+\mu_2 + \frac{4\beta \rho \mu_1}{\alpha} & \mu_1+\mu_2 \\
\mu_1+\mu_2 & \mu_1+\mu_2 + \frac{4\alpha \rho \mu_2}{\beta}
\end{pmatrix}
,\\
M &= P^{-1} Q.
\end{align*}
The solution of~\eqref{Y} is 
\[
Y(t) = \int_0^t e^{-M(t-s)} P^{-1}\e\, u(s) ds.
\]
Noting that 
$$
P^{-1}\e= 
\begin{pmatrix}
\alpha^2 \\
\beta^2
\end{pmatrix},
$$
we obtain that 
\begin{align*}
y_{1,T} &= \frac{\alpha^2}{2\gamma} \int_0^T
 \left( e^{(t-T)\lambda_1}(\gamma+\beta^2-\alpha^2+2) + 
 e^{(t-T)\lambda_2}(\gamma-\beta^2+\alpha^2-2) \right) u(t)\,dt, \\
y_{2,T} &= \frac{\beta^2}{2\gamma} \int_0^T
 \left( e^{(t-T)\lambda_1}(\gamma+\beta^2-\alpha^2-2) + 
 e^{(t-T)\lambda_2}(\gamma-\beta^2+\alpha^2+2) \right) u(t)\,dt,
\end{align*}
where $\gamma$, $\lambda_1$ and $\lambda_2$ are the quantities defined in~\eqref{gamma},~\eqref{lamb1}
and~\eqref{lamb2}, respectively.
Thus, we have obtained explicit formulas for $y_{1,T}$, $y_{2,T}$, and $y_{3,T}$
 (via the relation $2y_3=w-y_1-y_2$) in the Gaussian setting.

\bibliographystyle{plain}

\end{document}